\documentclass[conference]{IEEEtran}
\IEEEoverridecommandlockouts
\usepackage{cite}
\usepackage{amsmath,amssymb,amsfonts}
\usepackage{algorithmic}
\usepackage{graphicx}
\usepackage{textcomp}
\usepackage{xcolor}
\usepackage{amsthm}
\usepackage{overpic}
\usepackage{bm}
\usepackage{overpic}
\usepackage{float}
\usepackage{amsthm}

\theoremstyle{plain}

\newtheorem{lemma}{Lemma}

\newcommand\numberthis{\addtocounter{equation}{1}\tag{\theequation}}

\def\Htran{\mbox{\tiny $\mathrm{H}$}}
\def\Ttran{\mbox{\tiny $\mathrm{T}$}}
\def\CN{\mathcal{N}_{\mathbb{C}}} 

\begin{document}

\title{Optimizing Reconfigurable Intelligent Surfaces for Small Data Packets: A Subarray Approach
\thanks{This work was supported by the FFL18-0277 grant from the Swedish Foundation for Strategic Research.}
}

\author{\IEEEauthorblockN{Anders Enqvist$^*$, \"Ozlem Tu\u{g}fe Demir$^*$, Cicek Cavdar$^*$, and Emil Bj\"ornson$^{*\ddagger}$}
	\IEEEauthorblockA{{$^*$Department of Computer Science, KTH Royal Institute of Technology, Kista, Sweden
		} \\ {$^\ddagger$Department of Electrical Engineering, Link\"oping University, Link\"oping, Sweden
		} \\
		{Email: \{enqv, ozlemtd, cavdar, emilbjo\}@kth.se}
	}
	
}

\maketitle

\begin{abstract}
In this paper, we examine the energy consumption of a user equipment (UE) when it transmits a finite-sized data packet. The receiving base station (BS) controls a reconfigurable intelligent surface (RIS) that can be utilized to improve the channel conditions, if additional pilot signals are transmitted to configure the RIS. 
We derive a formula for the energy consumption taking both the pilot and data transmission powers into account.
By dividing the RIS into subarrays consisting of multiple RIS elements using the same reflection coefficient, the pilot overhead can be tuned to minimize the energy consumption while maintaining parts of the aperture gain. 
Our analytical results show that there exists an energy-minimizing subarray size. For small data blocks and when the channel conditions between the BS and UE are favorable compared to the path to the RIS, the energy consumption is minimized using large subarrays. When the channel conditions to the RIS are better and the data blocks are large, it is preferable to use fewer elements per subarray and potentially configure the  elements individually.
\end{abstract}


\section{Introduction}

Most wireless communication applications generate intermittent traffic \cite{asplund2020advanced}. When a user equipment (UE) should transmit a finite piece of data to a base station (BS), the packet size depends on which modulation-and-coding scheme is supported, which in turn is determined by the channel conditions. The channel used to be given by nature, but the advent of reconfigurable intelligent surfaces (RISs) makes it partially controllable \cite{8741198}. An RIS contains an array of sub-wavelength-sized elements that can be configured to reflect incident waves in desired ways \cite{8466374}, for example, as a beam towards the intended receiver. Recent works have described RIS experiments with over a thousand elements \cite{pei2021ris}. 
A practical obstacle is the additional pilot signaling overhead required to identify the preferred RIS configuration \cite{9311936}. This is particularly the case when only a finite-sized payload is to be transmitted.

The individual wireless links might not be capacity-constrained in 6G, thanks to enormous bandwidths, but rather limited by energy consumption. A large RIS can improve the energy-efficiency (EE) of the transmission, and beats traditional relays when it is sufficiently large \cite{8741198, 8888223,9497709}. These prior works use the capacity-to-power ratio as the EE metric, targeting continuous transmission with full buffers and deterministic channels.
The RIS-aided ergodic capacity of fading channels was studied in \cite{9244106,9333612}, where it was shown that particular numbers of RIS elements maximize the capacity or EE metric, respectively. A key assumption was that the pilot length equals the number of elements (plus one due to the BS-UE direct channel), thus the channel coherence limits how many RIS elements are useful.

In this paper, we take a new approach by considering the uplink transmission of a finite-sized data payload, as in practice \cite{asplund2020advanced}. 
The maximum EE is achieved by minimizing the energy consumption at the UE for a certain signal-to-noise ratio (SNR), with support by an RIS with a fixed number of elements, under the constraint that the receiver can decode the data.
 Our model groups several RIS elements into \emph{subarrays} that share the same reflection coefficient as in \cite{zheng2019intelligent}. However, in this paper its impact on the EE is analyzed. We show that the pilot length can be made equal to the number of subarrays (plus one) and selected to minimize energy consumption. We show that this solution is preferable to turning off a fraction of the RIS elements, as done in \cite{9244106,9333612}. We derive the optimal number of subarrays and show how it depends on the pathlosses, number of RIS elements, and payload size.

\section{System model}

We consider a system where a single-antenna UE transmits to a single-antenna BS. In the vicinity of the BS and UE, there is an RIS that can be utilized to improve the channel conditions. We assume that the BS configures and interacts with the RIS.
The question that arises is: when is it worth spending extra radio resources configuring the RIS, to minimize the total energy consumption of the UE for uplink transmission?

\begin{figure}[htbp] 
\centerline{\includegraphics[width=1\columnwidth]{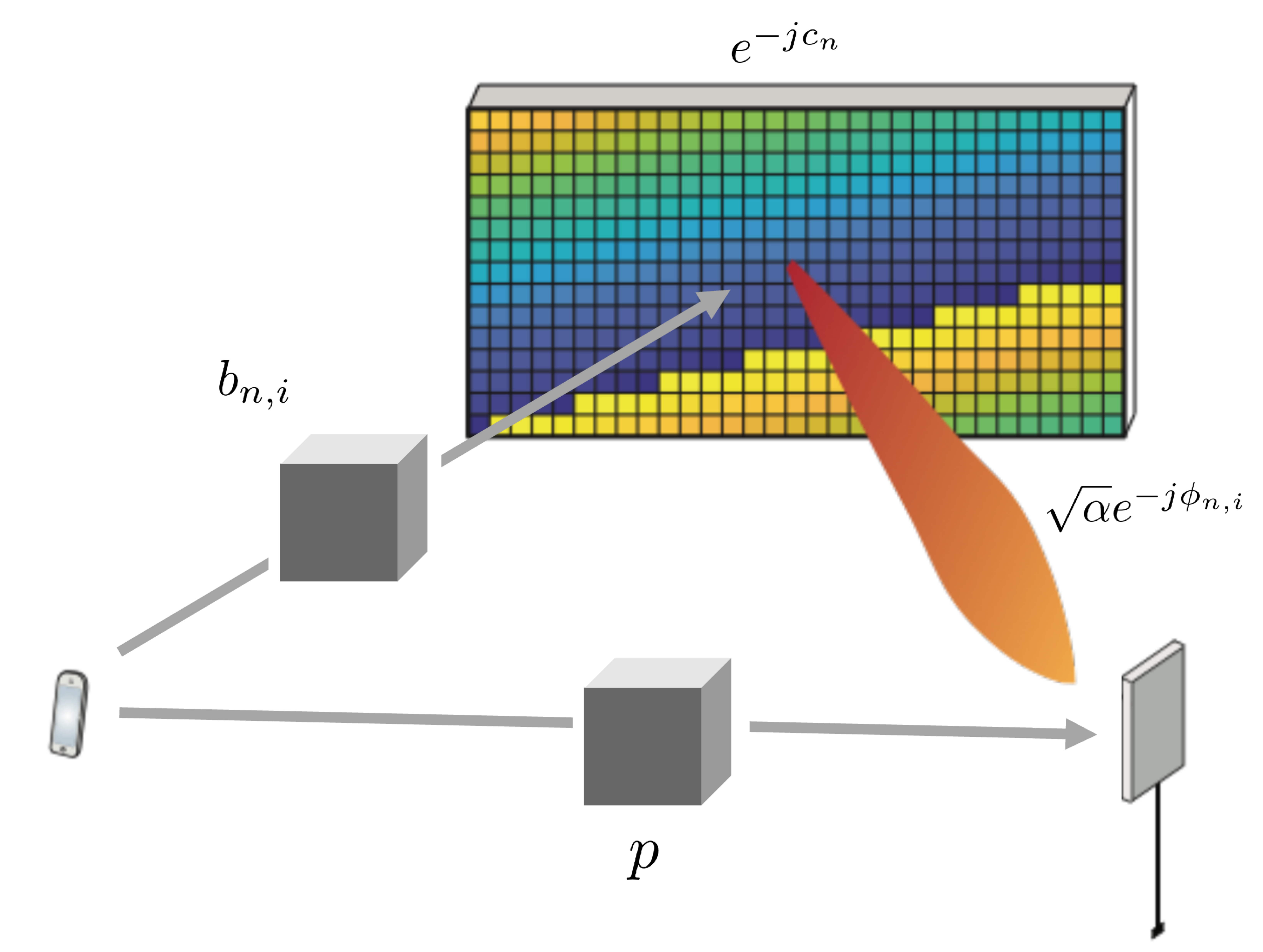}} \vspace{-3mm}
\caption{We consider a setup where a UE transmits to a BS through a direct path $p$ and a controllable RIS with $M$ elements. Two of the paths are in non-line-of-sight (NLoS) and one in LoS. The colors represent the phase-shifts of the RIS elements.}
\label{fig:system-model} \vspace{-2mm}
\end{figure}

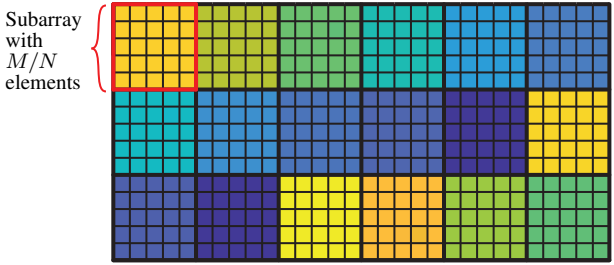
\begin{figure}[t!]
	\centering 
	\begin{overpic}[width=.9\columnwidth,tics=10]{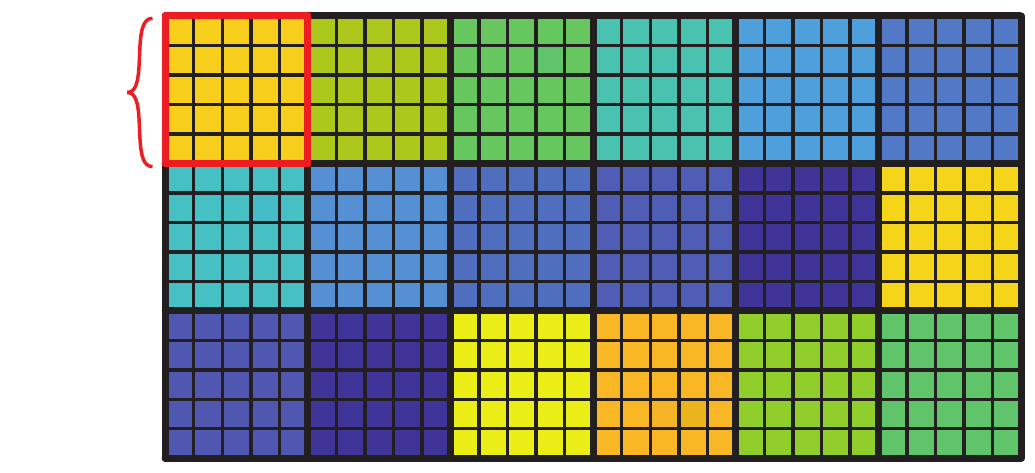}
		\put(-2,40.5){\footnotesize Subarray}
		\put(-2,37){\footnotesize  with}
		\put(-2,33.5){\footnotesize $M/N$}
		\put(-2,30){\footnotesize  elements}
\end{overpic}  \vspace{-2mm}
	\caption{An example setup with an RIS consisting of $M=450$ elements that are divided into $N=18$ subarrays, each with $M/N=25$ elements.}  \vspace{-5mm}
\label{fig:system-model2b}
\end{figure}

We model the problem as follows. The UE transmits the data signal $x$ with power $P_\mathrm{data}$, i.e., $\mathbb{E}\{\vert x \vert^2\}=P_\mathrm{data}$. The received signal at the BS is
\begin{equation} \label{eq:received-signal}
    y=gx+w,
\end{equation}
where $w\sim \CN (0,\sigma^2)$ is the complex Gaussian thermal receiver noise  and $g$ is the end-to-end channel response, which depends on the RIS configuration and is modeled as follows.

We consider the environment in Fig.~\ref{fig:system-model}, where the BS and RIS are  deployed to have a line-of-sight (LoS) channel between them, while all other channels feature Rayleigh fading.  
The direct path between the BS and UE is modeled as $p\sim \CN(0,\rho)$, where $\rho$ is the average channel gain.
There are also $M$ reflected paths via the $M$ RIS elements.
We assume the RIS elements are divided into $N$ subarrays, each consisting of $M/N$ elements, as illustrated in Fig.~\ref{fig:system-model2b}. 
The RIS assigns the same phase-shift $c_n$ to all the elements in subarray $n$, to limit the channel estimation overhead and reduce computational complexity.
We will treat $N$ as an optimization variable in this paper and stress that $N=M$ results in a conventional RIS with a different phase-shift at each RIS element. Determining the optimal number of elements $M$ is important from a design perspective however, once the RIS has been manufactured, it will not change. On a similar note, once the RIS and the BS have been deployed, the pathloss $\alpha$ between them will be a constant factor.
Since we have LoS paths from the RIS to the BS, all elements have a common gain of $\alpha$ but different phase-shifts. For the RIS element $i\in \{1,\ldots,M/N\}$ in subarray $n$, we let $\phi_{n,i}$ denote the phase-shift of the corresponding path, thus the channel coefficient is $\sqrt{\alpha} e^{-j\phi_{n,i}}$.
The channel from the UE to the RIS element $i$ in subarray $n$ is denoted and modeled as $b_{n,i} \sim \CN(0,\beta)$, where $\beta$ is the average channel gain.


In summary, the end-to-end channel response in \eqref{eq:received-signal} is
\begin{equation} \label{eq:end-to-end-channel}
g=
        p+
    \sum_{n=1}^{N} e^{-j c_n}
    \underbrace{\sum_{i=1}^{M/N} \sqrt{\alpha}e^{-j \phi_{n,i}} 
    b_{n,i}}_{=Z_n},
\end{equation}
where we define $Z_n=\sum_{i=1}^{M/N} \sqrt{\alpha}e^{-j \phi_{n,i}}b_{n,i}$ as the concatenated channel of subarray $n$.
We note that $Z_n \sim \CN(0,\frac{\alpha \beta M}{N})$ and thus $|Z_n|\sim \mathrm{Rayleigh}(\sqrt{
\frac{\alpha \beta M}{2N}})$ is Rayleigh distributed.
The term $e^{-j c_n}$ is the common phase-shift of the elements in subarray $n$.\footnote{We have not explicitly included the absorption losses in the RIS elements since these can be absorbed into the channel gains.}
By using \eqref{eq:end-to-end-channel}, we can rewrite \eqref{eq:received-signal} as
\begin{equation} \label{eq:received-signal2}
y=\left(p+\sum_{n=1}^{N} Z_n e^{-j c_n} \right)x+w.
\end{equation}


\subsection{Signal-to-noise ratio in the data transmission}

When the channel in \eqref{eq:received-signal2} is used for data transmission, the $\mathrm{SNR}$ for the configuration $c_1,\ldots,c_N$ over $N$ subarrays is 
\newline
\begin{equation} \label{eq:SNR}
\mathrm{SNR}(c_1,\ldots,c_N) = \ \frac{P_\mathrm{data}}{\sigma^2}\left\vert
p+ \sum_{n=1}^{N} Z_n  e^{-j c_n} \right\vert^2
\end{equation}
\newline
and is maximized by $c_n=\arg(Z_n)-\arg(p)$, which gives
\begin{equation} \label{eq:max-SNR}
    \overline{\mathrm{SNR}} = \max_{c_1,\ldots,c_N} \mathrm{SNR}(c_1,\ldots,c_N) = \!
\frac{P_\mathrm{data}}{\sigma^2} \! \left(|p|+ \sum_{n=1}^{N} |Z_n|\right)^2.
\end{equation}

When transmitting a finite-sized data packet, the average $\mathrm{SNR}$ determines the likelihood of successful reception (along with the modulation and coding scheme).
Before we derive the exact expression for the average $\mathrm{SNR}$, we will consider a tight lower bound on it, which will be useful when extracting analytical insights.
A lower bound $\overline{\mathrm{SNR}}_{\rm av,low}$ on the average $\mathrm{SNR}$ can be obtained using Jensen's inequality as
\begin{align}
\label{eq:Jensen}
    \mathbb{E} \left\{\overline{\mathrm{SNR}}\right\} & \geq \overline{\mathrm{SNR}}_{\rm av,low} = \frac{P_\mathrm{data}}{\sigma^2} \left(\mathbb{E}\left\{|p|+ \sum_{n=1}^{N} |Z_n|\right\}\right)^2 \nonumber \\
    &= \frac{P_\mathrm{data}}{\sigma^2} \left(\frac{\sqrt{\pi}}{2}\sqrt{\rho}+  N  \frac{\sqrt{\pi}}{2}\sqrt{
\frac{\alpha \beta M}{N}}  \right)^2 \nonumber \\
&=\frac{\pi P_\mathrm{data}}{4\sigma^2} 
\left(\sqrt{\rho}+\sqrt{\alpha \beta M N}  \right)^2
\end{align}
where the equality in the second line follows from the fact that 
$E\{|p|\}=\frac{\sqrt{\pi}}{2}\sqrt{\rho}$ and $E\left\{|Z_n|\right\} = \frac{\sqrt{\pi}}{2}\sqrt{
\frac{\alpha \beta M}{N}}$.

The following lemma gives the exact average $\mathrm{SNR}$.

\begin{lemma}
With an optimized RIS configuration, the average of the $\mathrm{SNR}$ in \eqref{eq:max-SNR} is
\begin{equation} \label{eq:average-SNR}
\mathbb{E}\left\{\overline{\mathrm{SNR}}\right\} = \overline{\mathrm{SNR}}_{\rm av,low} + \frac{ P_\mathrm{data}}{ \sigma^2}\left(1-\frac{\pi}{4}\right)\left(\rho+\alpha\beta M\right) .
\end{equation}
\end{lemma}
\begin{IEEEproof}
By expanding the square in \eqref{eq:max-SNR}, we obtain 
\begin{align}
\mathbb{E} \left\{\overline{\mathrm{SNR}}\right\}=& 
\frac{P_\mathrm{data}}{\sigma^2}\Bigg(\mathbb{E}\left\{|p|^2\right\}+2\mathbb{E}\left\{|p|\right\}\sum_{n=1}^{N} \mathbb{E}\left\{|Z_n|\right\} \nonumber \\
&+\sum_{n=1}^{N} \sum_{m=1}^{N}\mathbb{E}\left\{|Z_n||Z_m|\right\}\Bigg).  
\end{align}
By using the independence of $Z_n$ from $p$ and $Z_m$, for $n\neq m$, $\mathbb{E} \left\{\overline{\mathrm{SNR}}\right\}$ can be simplified as
\begin{align}
&\frac{P_\mathrm{data}}{\sigma^2} \Bigg(\rho+\frac{\pi}{2}\sqrt{\rho\alpha\beta M N} + N \frac{\alpha \beta M}{N}
 \nonumber \\
&+N(N-1)\left( \frac{\sqrt{\pi}}{2} \sqrt{\frac{\alpha\beta M}{N}}\right)^2 \Bigg) \nonumber \\
&= \frac{\pi P_\mathrm{data}}{4 \sigma^2} \Bigg(
\left(\sqrt{\rho}+\sqrt{\alpha \beta M N}\right)^2
+\left(\frac{4}{\pi}-1\right) \!(\rho+\alpha \beta M) \!
\Bigg).    
\numberthis
\label{eq:exactSNR}
\end{align}
We finally obtain \eqref{eq:average-SNR} by identifying $\overline{\mathrm{SNR}}_{\rm av,low}$ from \eqref{eq:Jensen}.
\end{IEEEproof}

The exact average $\mathrm{SNR}$ in \eqref{eq:average-SNR} is an increasing function of the channel gains $\rho$, $\alpha$, $\beta$, the number of RIS elements $M$, and the number of subarrays $N$.
We further notice that \eqref{eq:average-SNR} is a summation of its lower bound in \eqref{eq:Jensen} and a term that is independent of $N$. 
The latter term is often negligible.
%
To demonstrate this, Fig.~\ref{fig:ESNR} compares the exact $\mathrm{SNR}$ expression in \eqref{eq:average-SNR} and its lower bound in (\ref{eq:Jensen}). We consider a setup where the channel gains are
$\alpha=-80$\,dB, $\beta=-60$\,dB, and $\rho=-95$\,dB. The RIS consists of $M=1024$ elements, which can be divided into subarrays of sizes $N=2^K$ for $K=0,1, \dotsc ,10$.
The ``transmit SNR'' is $P_\mathrm{data}/\sigma^2=104$\,dB. Also shown is a comparison to not using subarrays but instead $N$ individually configured elements. In this baseline case, where $M-N$ elements are turned off as in prior work, we obtain the average maximum SNR from \eqref{eq:exactSNR} as
\begin{align}
&\mathbb{E}\left\{\overline{\mathrm{SNR}_b}\right\}= \nonumber \\
&\hspace{2mm}\frac{\pi P_\mathrm{data}}{4 \sigma^2} \Bigg(
\left(\sqrt{\rho}+N\sqrt{\alpha \beta}\right)^2
+\left(\frac{4}{\pi}-1\right) \!(\rho+\alpha \beta N) \!
\Bigg).
\end{align}
It is clearly seen that
$\mathbb{E}\left\{\overline{\mathrm{SNR_b}}\right\} \leq \mathbb{E}\left\{\overline{\mathrm{SNR}}\right\}$
with equality only if $N=M$. This is also observed in Fig.~\ref{fig:ESNR} which shows how the subarrays maintain the aperture gain determined by the RIS size, even if the beamforming gain is reduced by sharing configurations between elements. The stars ($*$), which overlap with the formulas in Fig.~\ref{fig:ESNR}, are obtained from 15,000 Monte Carlo realizations and validate the accuracy of the results.
\begin{figure}[t!]
\includegraphics[width=1\columnwidth]{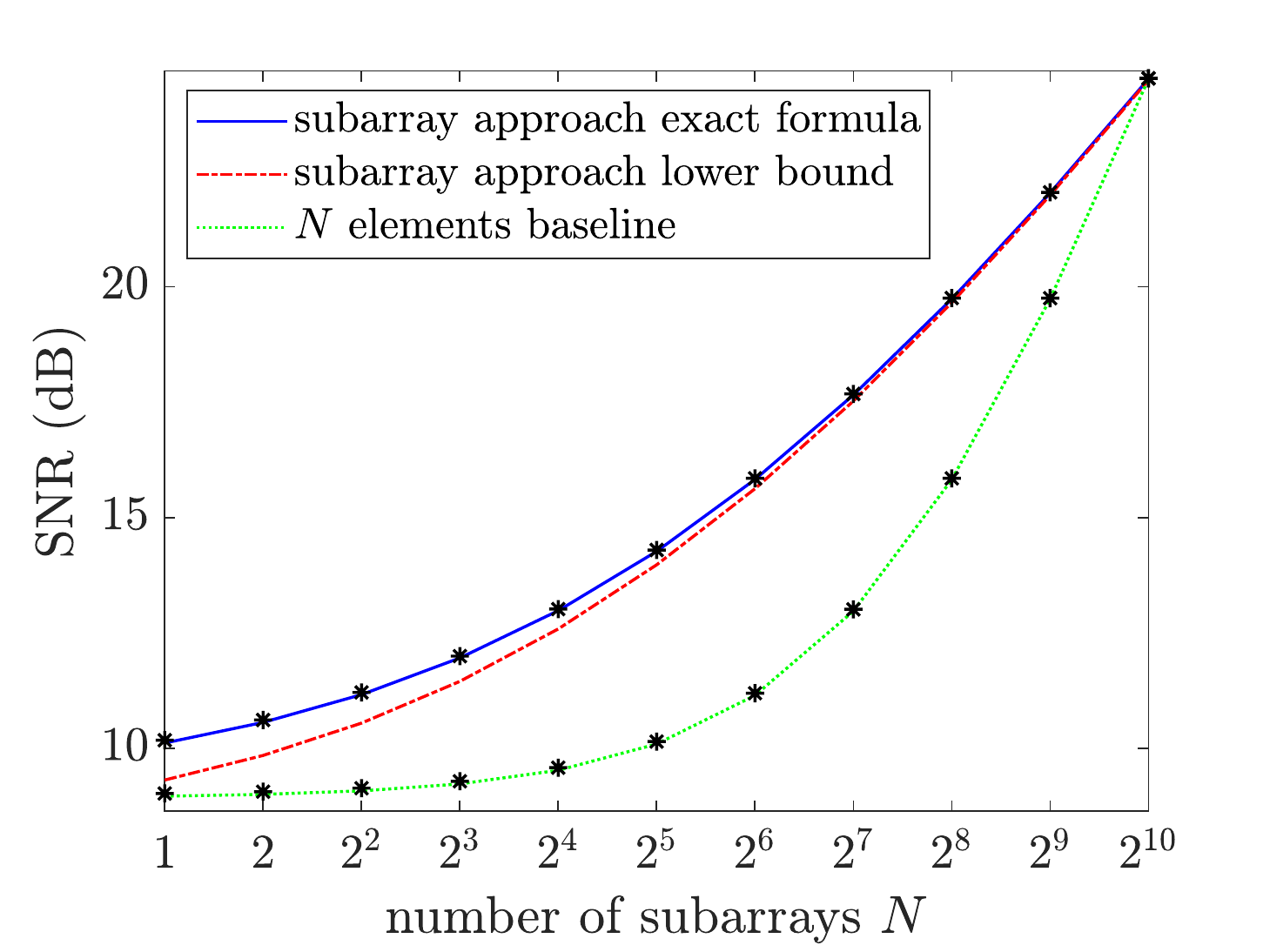} 
\caption{$\mathrm{SNR}$ versus $N$ for $\alpha=-80$\,dB, $\beta=-60$\,dB, and $\rho=-95$\,dB for $M=1024$ RIS elements being divided into subarray sizes of $N=2^K$ for $K=0,1, \dotsc ,10$. The baseline with only $N$ individually controlled elements is also shown.
} \label{fig:ESNR}
\end{figure}

\subsection{Signal-to-noise ratio in the pilot transmission}


To compute the RIS configuration that maximizes the $\mathrm{SNR}$ in \eqref{eq:max-SNR}, we need to learn the current channel realization. 
When using $N$ subarrays, there are effectively $N+1$ channel coefficients in  \eqref{eq:received-signal2} to estimate; one for each subarray with RIS elements plus the direct path. Let 
\begin{align} \label{eq:channel-h}
\textbf{h} =    \begin{bmatrix} p & Z_1 & \hdots & Z_{N} \end{bmatrix}^{\Ttran} \in \mathbb{C}^{N+1}
\end{align}
denote the vector containing these channel coefficients from \eqref{eq:received-signal2}.
The benefit of using subarrays is apparent in that the number of coefficients is lower than the total number of RIS elements.
Since there are $N+1$ unknown coefficients, we assume the UE sends a predefined pilot sequence of length $N+1$ in order to obtain sufficient degrees-of-freedom to estimate the individual coefficients  in $\textbf{h}$. 
Without loss of generality, the pilot sequence is constant $\sqrt{P_\mathrm{pilot}}[1 \ \ldots \ 1 ]^{\Ttran}\in \mathbb{C}^{N+1}$, where $P_\mathrm{pilot}$ is the transmit power.
During the pilot transmission, the RIS changes its configuration according to a structured deterministic phase-shift pattern. Let $\bm{\psi}_i\in \mathbb{C}^{N+1}$ be the vector with $1$ as its first entry (representing the direct path) and the configured phase-shifts $e^{-j\psi_{i,n}}$ for subarrays $n=1,\ldots,N$ during the $i$th pilot  transmission as the other entries. This vector is given as
\begin{align}
   \bm{\psi}_i = \begin{bmatrix} 1 & e^{-j\psi_{i,1}} & \hdots & e^{-j\psi_{i,N}} \end{bmatrix}^{\Ttran}.
\end{align}
 Then, the received signal at the BS is given in vector form as
\begin{align}
    \textbf{y}_\mathrm{pilot} = \sqrt{P_\mathrm{pilot}}\underbrace{\begin{bmatrix} \bm{\psi}_1^{\Ttran} \\ \vdots \\ \bm{\psi}_{N+1}^{\Ttran}  \end{bmatrix}}_{\triangleq \boldsymbol{\Psi}}\textbf{h}+ \textbf{w}_\mathrm{pilot}
\end{align}
where $\textbf{w}_\mathrm{pilot}\sim \CN(\textbf{0},\sigma^2\textbf{I}_{N+1})$ is the thermal receiver noise. If we select the matrix $\boldsymbol{\Psi}$ as any scaled unitary matrix (e.g., a discrete Fourier transform matrix), then $\boldsymbol{\Psi}^{\Htran}\boldsymbol{\Psi}=(N+1)\textbf{I}_{N+1}$ and the BS can estimate the individual paths as
\begin{align} 
\label{eq:unitarypilot}
 \frac{\boldsymbol{\Psi}^{\Htran}\textbf{y}_\mathrm{pilot}}{\sqrt{N+1}}  = \sqrt{(N+1)P_\mathrm{pilot}}\textbf{h} +  \frac{\boldsymbol{\Psi}^{\Htran}\textbf{w}_\mathrm{pilot}}{\sqrt{N+1}}   
\end{align}
where the effective noise component has the same distribution as $\textbf{w}_\mathrm{pilot}$ since $\boldsymbol{\Psi}^{\Htran}/\sqrt{N+1}$ is a unitary matrix. Recalling the entries of the channel vector $\textbf{h}$ from \eqref{eq:channel-h}, we have different pilot $\mathrm{SNR}$ for the BS-UE direct path and the paths through the $N$ subarrays of RIS. Using the channel statistics derived before, the pilot $\mathrm{SNR}$ for the direct channel and the cascaded RIS channels are given as 
\begin{align} \label{eq:SNR_pilot}
\mathrm{SNR}_\mathrm{pilot} = \begin{cases} \frac{(N+1)P_\mathrm{pilot}\rho}{\sigma^2} & \text{for the direct channel}, \\
\frac{(N+1)P_\mathrm{pilot}\alpha\beta M}{N\sigma^2} & \text{for the RIS channels}.\end{cases}
\end{align}
In configuring the RIS phase-shifts, if the direct channel gain $\rho$ is larger than each cascaded RIS subsurface channel gain $\alpha\beta M/N$, a value of $P_\mathrm{pilot}$ that is selected according to the minimum pilot $\mathrm{SNR}$ requirement of the RIS channels will ensure that the direct channel pilot $\mathrm{SNR}$ is also above the required threshold. When the direct channel gain is relatively low, aligning the RIS phase-shifts according to $p$ will not have a significant effect on the maximum data $\mathrm{SNR}$ in \eqref{eq:max-SNR}. Again in this case, we should focus on the pilot $\mathrm{SNR}$ of the RIS channels in selecting the pilot transmit power $P_\mathrm{pilot}$ to guarantee accurate channel estimates. Motivated by this, in our design problem, we will focus on the pilot $\mathrm{SNR}$ of the RIS channels that is given in the second line of \eqref{eq:SNR_pilot}.


\subsection{Required pilot $\mathrm{SNR}$ to achieve maximal data $\mathrm{SNR}$}
\label{subsec:required-pilot-SNR}

When transmitting the pilot signals, they will be received with an $\mathrm{SNR}$ value of $\gamma_p$ ($\mathrm{SNR}_\mathrm{pilot}$ in \eqref{eq:SNR_pilot}) depending on the allocated power. If the value of $\gamma_p$ is not high enough, it will lead to the BS not being able to adequately compute the optimal phase-shifts needed to reach the maximum data $\mathrm{SNR}$ in (\ref{eq:max-SNR}). The performance loss is defined as the ratio between the achieved $\mathrm{SNR}$ in \eqref{eq:SNR} with the phase-shifts that are adjusted according to the channel estimates and the maximum in (\ref{eq:max-SNR}).

We used (\ref{eq:unitarypilot}) when calculating $\frac{(4)}{(5)}$ for different $\gamma_p$ values and noticed that
as $N$ increases the situation becomes worse, therefore we consider $N=M$.
The power that is allocated to the pilot symbols increases the achievable data $\mathrm{SNR}$ albeit at an additional energy consumption. 
However, achieving the target $\mathrm{SNR}$ in the pilot transmission is only possible if the channel gains are known beforehand. Since the purpose of the pilot is to learn the channel it will \emph{not} be treated as an optimization variable. In Fig.~\ref{fig:csi} we have simulated the fraction of maximal achievable data $\mathrm{SNR}$ $\gamma_d$ for $M=1024$ elements and varying pathloss values. If the direct path is weak, in the sense that $\rho/(\alpha \beta)$ is 100 or 0, then we need a pilot $\mathrm{SNR}$ of 20\,dB to get close to the perfect channel state information (CSI) case. We will use this value in later simulations. In the remainder we assume the large-scale fading parameters $\alpha, \beta ,\rho$ are known beforehand, so we can select the target value $\gamma_p$ to achieve a negligible performance loss in the data transmission phase.

\begin{figure}[t!]
\includegraphics[width=1\columnwidth]{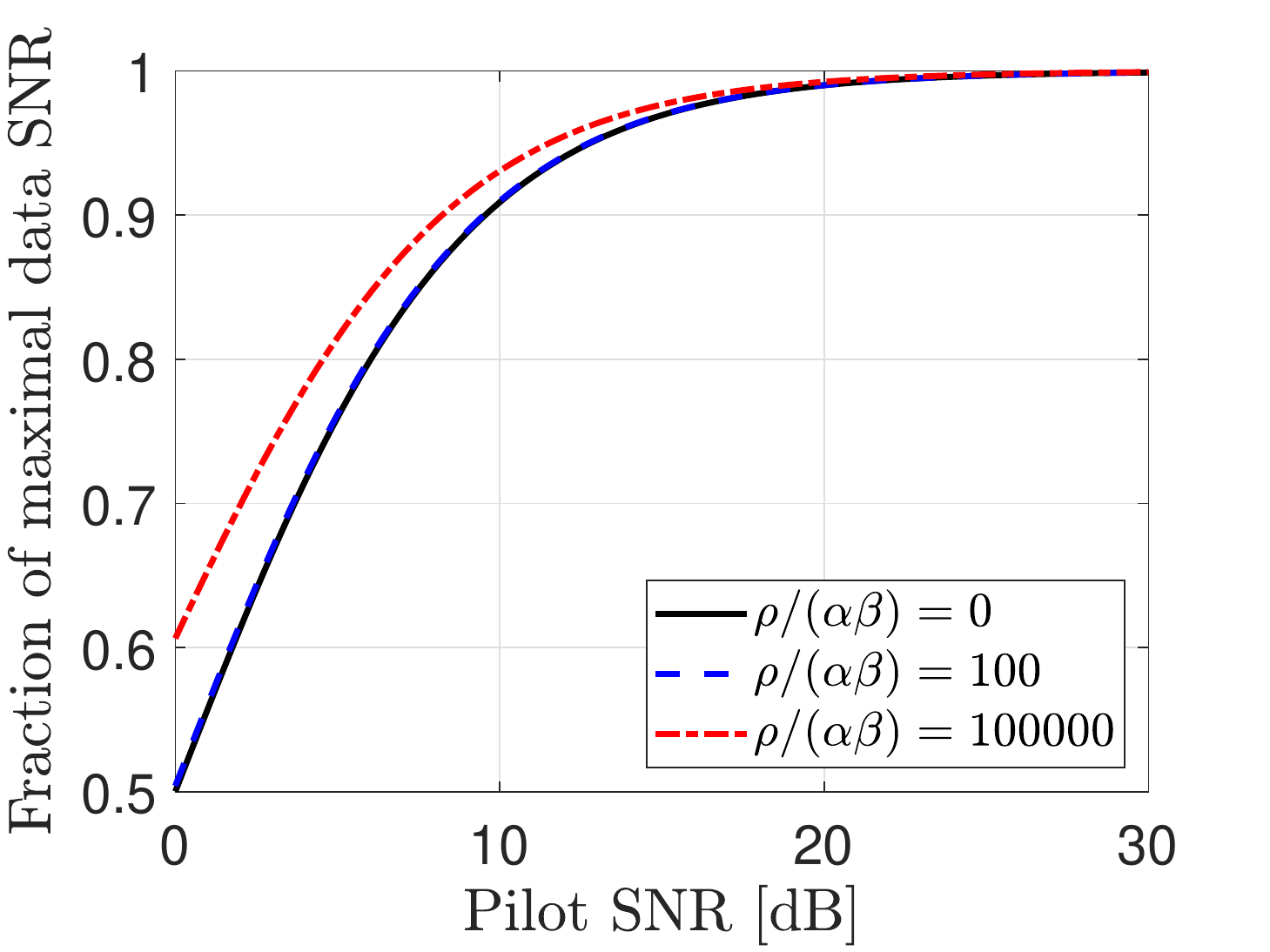}
\caption{Relative performance loss in data $\mathrm{SNR}$ $\gamma_d$ for $N=1024$ individually configured elements as a function of pilot $\mathrm{SNR}$ $\gamma_p$.} \label{fig:csi}
\end{figure}


\section{Optimizing the Subarray Size}

We consider the scenario where the UE wants to transmit a data packet with a finite-sized information content. 
This payload data is mapped to $L$ complex modulation symbols, based on some joint coding and modulation scheme designed to work at an $\mathrm{SNR}$ of $\gamma_d$. The value of $L$ greatly depends on how many bits the BS has scheduled for the UE together with the modulation and coding scheme. 
As a reference, in 5G OFDM-based systems, if a UE is scheduled on a single resource block during one time slot it would transmit $12 \cdot 14 = 168$ symbols. If the UE is scheduled on the whole carrier, consisting of 273 resource blocks in the mid-band, during 3 time slots it would transmit $3 \cdot 12 \cdot 273 \cdot 14 = 137,592$ symbols. Depending on the channel coherence time, the pilot measurement may be valid for multiple time slots \cite{dahlman20205g}. This leads to the possibility of $L$ attaining values in the order of $L=200$ to $L=150,000$. We will consider both cases herein.

We want to minimize the energy consumption at the UE while ensuring successful communication.
The energy consumption $E(N)$ of transmitting the $N+1$ pilots during channel estimation (with $N$ subarrays) and the $L$ data symbols is
\begin{equation} \label{eq:energy-consumption1}
    E(N)= (N+1)\frac{P_\mathrm{pilot}}{B}+L \frac{P_\mathrm{data}}{B},
\end{equation}
where $B$ is the symbol rate.

Assuming that a pilot $\mathrm{SNR}$ of $\gamma_p$ is needed for the channel estimation (see Section~\ref{subsec:required-pilot-SNR}), we obtain from the second line of \eqref{eq:SNR_pilot} and \eqref{eq:average-SNR} the required transmission powers as
\begin{align}
 P_\mathrm{pilot}=&\frac{\sigma^2 \gamma_p}{\alpha \beta M}\frac{N}{N+1}, \\
P_\mathrm{data}= &\frac{\sigma^2 \gamma_d}{\frac{\pi}{4} \left(\sqrt{\rho}+ \sqrt{\alpha \beta M N}\right)^2+\left(1-\frac{\pi}{4}\right)(\rho+\alpha \beta M)
}.
\end{align}

Inserting the above equations into \eqref{eq:energy-consumption1}, the energy consumption can be written as
\begin{align*}
\label{eq:exact}
E(N) &=\frac{L}{B} \frac{\sigma^2 \gamma_d}{\frac{\pi}{4}\left(\sqrt{\rho}+ \sqrt{\alpha \beta M N}\right)^2
+\left(1-\frac{\pi}{4}\right)(\rho+\alpha \beta M)}\\
&+\frac{N}{B} \frac{\sigma^2 \gamma_p}{\alpha \beta M}.\numberthis{}
\end{align*}
%
This equation shows that the energy consumption reduces as the number of elements $M$ in the RIS increases for a fixed subarray number $N$. Having more elements therefore gives both increased $\mathrm{SNR}$ and reduced energy consumption.


A subarray size that strikes an optimal balance between the $\mathrm{SNR}$ and the the pilot overhead may be found. 
To determine the subarray size $N$ that minimizes the energy consumption we may solve the following optimization problem:


\begin{equation} \label{eq:optimization}
\begin{aligned}
 \underset{N}{\textrm{minimize}} \quad & E(N)\\
\textrm{subject to} \quad & N\in \{1,\ldots,M\}, \quad \frac{M}{N} \in \mathbb{Z}  .
\end{aligned}
\end{equation}
The optimization is based on the long-term information ($\alpha, \beta, \rho$) since the exact channel realizations are unknown until after the pilots have been transmitted.

The first and second order derivatives of the energy consumption when inserting $\overline{\mathrm{SNR}}_{\rm av,low}$ from (\ref{eq:Jensen}) in the energy consumption \eqref{eq:exact} are: \vspace{-3mm}

\begin{equation} \label{eq:first-derivative}
    E'(N)=\frac{\sigma^2 \gamma_p}{B\alpha \beta M} - \frac{\sigma^2 \gamma_d L}{B}
    \frac{\sqrt{\alpha \beta M}}{\frac{\pi}{4}\left(\sqrt{\rho}+\sqrt{\alpha \beta M N} \right)^3\sqrt{N}},
\end{equation}

\begin{equation}
    E''(N)=\frac{\sigma^2 \gamma_d L}{B}\frac{\sqrt{\alpha \beta M}(\sqrt{\rho}+4\sqrt{\alpha \beta M N})}{\frac{\pi}{2} (\sqrt{\rho}+\sqrt{\alpha \beta M N})^4N^{3/2}}.
\end{equation}

It can be verified that $E''(N)>0$ for every $N>0$.  This can also be shown for the exact expression. 
Hence, a unique positive real-valued solution $N^{\star}$ to the equation $E'(N^{\star})=0$ which minimizes $E(N)$ exists  as we show in Lemma~\ref{lem2} below.
It is only when $1\leq N^{\star}\leq M$, however, that the solution is viable. When all subarrays consist of the same number of elements then $M/N$ must also be an integer. The following lemma formalizes this result and proposes a way to find the optimal integer $N$ to the problem in \eqref{eq:optimization}.


\begin{lemma} \label{lem2}
The optimal energy-efficient solution $N^{\star}$ that minimizes $E(N)$ exists and is obtained by finding the unique positive real-valued root of $E^{\prime}(N)=0$ via a bisection search. When $N$ is restricted to be a positive integer such that also $M/N$ is an integer, the optimal $N$ can be found by evaluating $E(N)$ for the closest smaller and larger feasible $N$ values to $N^{\star}$ and selecting the one giving the minimum $E(N)$.
 \end{lemma}
\begin{proof}
The optimal $N^{\star}$ that minimizes $E(N)$ can be found by equating the derivative of $E(N)$ to zero, i.e., $E'(N)=0$. This corresponds to equating the first term of $E'(N)$ in \eqref{eq:first-derivative}, to the second term, which is decreasing with $N$. Moreover, the second term approaches infinity as $N$ goes to zero and approaches zero as $N \to \infty$. Hence, there is a unique $N^{\star}>0$ that satisfies $E'(N)=0$ which can be quickly found by a bisection search. 
Moreover, since $E(N)$ is a convex function, the optimal $N$ that is both integer and divides $M$ can be found by checking the objective function $E(N)$ for the closest smaller and larger $N$ values that are feasible.
Note that these claims are valid also for the exact expression of $E(N)$. 
\end{proof}

As a general rule, once the UE transmits data, the configured RIS improves the $\mathrm{SNR}$ in the data transmission but the actual increase depends on the strength of the different channel paths. A strong direct channel gain $\rho$ means a high value of $\gamma_d$ is possible for similar transmission powers for the pilot and data signals.
This phenomenon is illustrated in Fig.~\ref{fig:pilotvsdata}, which shows the data transmission power versus pilot transmission power needed to reach the thresholds for different channel conditions. The parameters are $\alpha=-80$\,dB, $\beta=-50$\,dB, $\gamma_p=20$\,dB, $M=1024$ RIS elements, and $L=200$ data symbols. The subarray size increases from $N=1$ to $N=1024$ as we move to the right. The point corresponding to the optimal viable energy-efficient subarray size is encircled on each curve.

We consider two scenarios in Fig.~\ref{fig:E_weak}: one in which the UE has a strong direct path and one where the direct path is weak. The cascaded path between the BS and the UE via the RIS is the same in both cases. We analyze $E(N)$ for varying payload sizes as well. Fig.~\ref{fig:E_weak} shows that 
$N^{\star}$ 
increases as the number of data symbols $L$ to transmit increases. $N^{\star}$ is also increased when the direct channel gain $\rho$ becomes smaller.


\begin{figure}[t!]

\begin{overpic}[width=1\columnwidth, trim={0 0 0 5mm}]{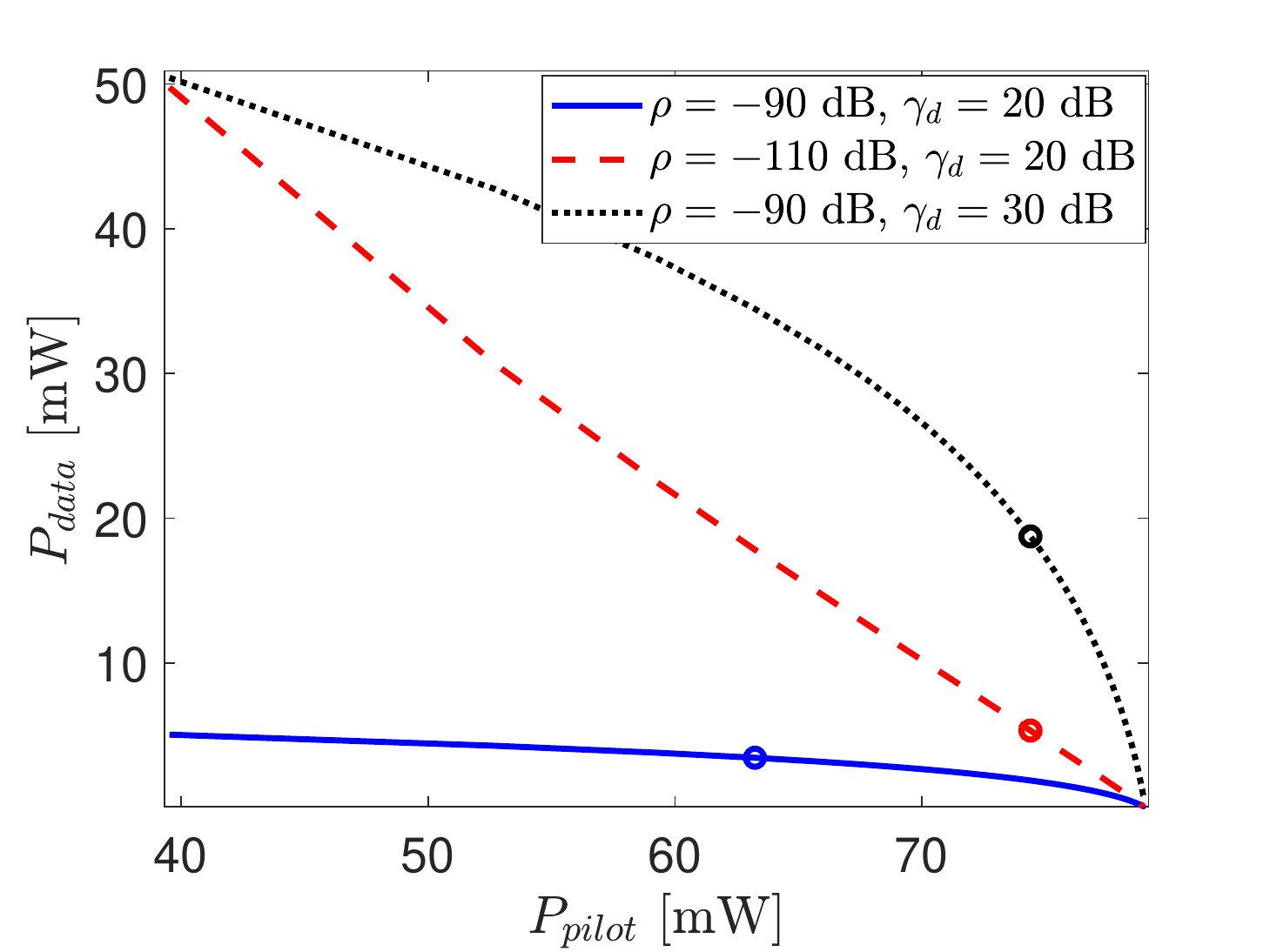} 
\put(59.5,21){\vector(0,-1){4}}
\put(54,21.5){\footnotesize $N=4$}
\put(81.2,23){\vector(0,-1){4}}
\put(81.2,27){\vector(0,1){4}}
\put(74,24){\footnotesize $N=16$}
\put(17.8,22){\vector(-1,-1){4}}
\put(16,23){\footnotesize $N=1$}
\end{overpic}
\vspace{-6mm}
\caption{Data transmission power versus pilot transmission power for varying direct path gain $\rho$ and data $\mathrm{SNR}$ target $\gamma_d$. The optimal viable energy-efficient subarray size is encircled.} 
\label{fig:pilotvsdata}
\end{figure}
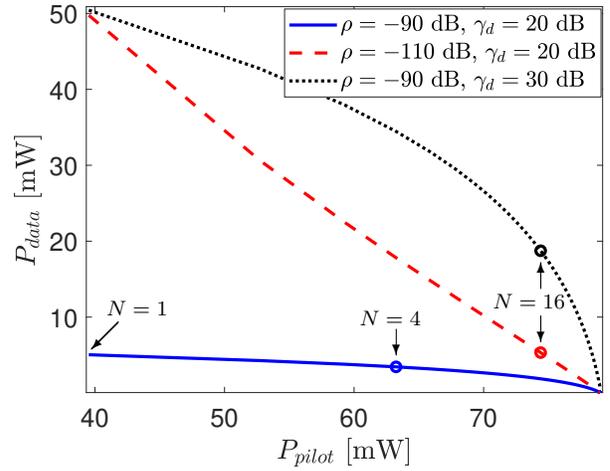


\subsection{Insights from the SNR lower bound}

When studying the behaviour of the energy consumption one may use the lower bound since its behaviour is similar to the exact average $\mathrm{SNR}$ (as shown in Fig.~\ref{fig:ESNR}) but the expression is simpler. The minimum to $E(N)$ is found either at $N=1$, $N=M$, or an interior point. The existence of an interior point solution can be determined by evaluating the first derivative in \eqref{eq:first-derivative}. There are three distinct cases for a convex function:

\begin{enumerate}
\item If $E'(1)<0$ and $E'(M)>0,$ the viable optimum is an interior point. This point can be quickly found utilizing the bisection search method.
\item If $E'(1)<0$ and $E'(M)<0$. The optimal solution is individually controlled RIS elements, i.e., $N=M$.
\item If $E'(1)>0$ and $E'(M)>0$. The optimal solution is to bundle all RIS elements into one subarray, i.e., $N=1$.
\end{enumerate}



We can further analyze when these cases occur.
The condition $E'(1)<0$ for the lower bound is equivalent to
\begin{equation}
    \sqrt{\frac{\rho}{\alpha \beta M}}< \left( \left(\frac{4\gamma_d L }{\pi \gamma_p}\right)^{1/3}-1
    \right).
    \label{eq:ineq1}
\end{equation}
The (exact) energy consumption $E(N)$ from (\ref{eq:exact}) is seen in Fig.~\ref{fig:E_weak}. For a strong direct path $\rho=-90$\,dB and small payload of $L=200$ bits, the inequality (\ref{eq:ineq1}) does not hold. 



\begin{figure}[t!]
\centerline{\includegraphics[width=1\columnwidth]{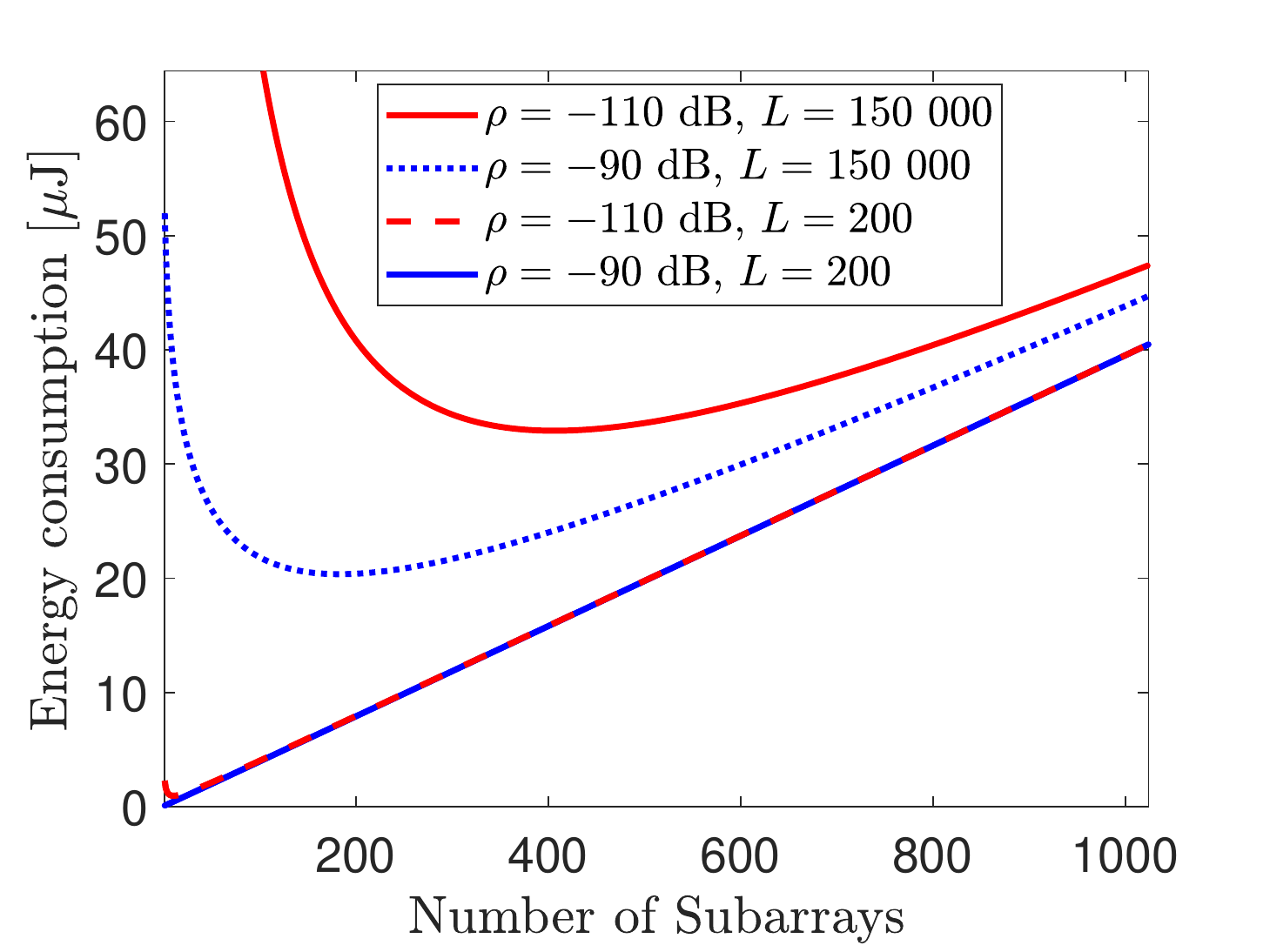}}
\caption{Energy consumption $E(N)$ for $\alpha=-80$\,dB, $\beta=-60$\,dB, $M=1024$ RIS elements when transmitting $L$ symbols over a channel with varying direct pathloss.  $N^{\star}$ corresponds to the minimum point on each curve.}  \label{fig:E_weak}
\end{figure}

The condition $E'(M)>0$ for the lower bound can be shown to be equivalent to
\begin{equation}
   \sqrt{\frac{\rho}{\alpha \beta M}}> \left( \left(\frac{4\gamma_d L}{\pi\gamma_p\sqrt{M} }\right)^{1/3}-\sqrt{M}
    \right).
    \label{eq:ineq2}
\end{equation}

The term $\sqrt{\frac{\rho}{\alpha \beta M}}$ represents the average quotient between the direct path and the reflected paths signal strength. If it is large, then (\ref{eq:ineq1}) won't hold, because the direct path will be efficient enough so that the RIS elements will not aid the communication sufficiently to warrant the increased energy cost of transmitting the RIS pilots, leading to employing only a single subarray as being the most energy-efficient choice.

In contrast, if $\sqrt{\frac{\rho}{\alpha \beta M}}$ is very small, the direct path has a negligible contribution to the received $\mathrm{SNR}$. If the inequality in \eqref{eq:ineq2} does not hold, it means that the system is starved for received signal power to the point that the most energy-efficient setup uses all the RIS elements steered individually. Note that the right-hand side of (\ref{eq:ineq2}) can attain negative values and that these inequalities were derived using the derivative in (\ref{eq:first-derivative}) for the lower bound expression. 

\section{Conclusion}
Our study showed that the pilot overhead is a limiting factor when minimizing the energy consumption for transmitting a finite-sized data packet over an RIS-assisted SISO (single-input single-output) channel. However, rather than turning off RIS elements to reduce the dimensionality, we have shown that it is preferable to group the RIS elements into subarrays with a common reflection coefficient because we can then maintain the aperture gain of the RIS (i.e., how much signal energy it collects), even if the beamforming gain is reduced.
We showed that the proposed dynamic grouping of elements into subarrays can greatly reduce the energy consumption compared to the baseline where each RIS element is configured individually. The optimal subarray configuration for a given UE depends on the relative strength of the channel to the RIS compared to the direct path and the payload size. Small payloads and a relatively weak path to the RIS favored large subarrays to reduce the pilot overhead, while the opposite was true for large payloads and strong paths to the RIS. Furthermore, the power that was allocated to the pilot sequence also affected both the energy consumption of the UE as well as the achievable date rate. If the RIS switches between aiding different UEs, the optimal subarray size will also change between UEs.

Since the subarray approach provides decisive energy reductions, it will be interesting to apply it to more advanced setups that incorporate multiple antennas at the BS and UE, as well as multiple RISs.
One could also minimize metrics that involve the energy consumption at the RIS and BS.

\bibliographystyle{IEEEtran}

\bibliography{Referenser}

\end{document}